\documentclass[jmp,preprint]{revtex4-1}
\draft

\usepackage{amsfonts, amsmath, amssymb, amsthm, graphicx,mathtools }

\usepackage{natbib}
\newtheorem{theorem}{Theorem}
\newtheorem*{theorem*}{Theorem}

\newtheorem{Lemma}[theorem]{Lemma}

\newtheorem{proposition}[theorem]{Proposition}

\newtheorem{result}[theorem]{Result}
\usepackage{bbold}

\begin{document}
\title{Killing of Transport in Lattices driven by Local Quantum Stochastic Dynamics. } 
\author{Beno\^it Descamps}%

\email{benoit.descamps@univie.ac.at}
\affiliation{Faculty of Physics, University of Vienna, Austria}%

\date{\today}

\begin{abstract}
Systems with local dynamics are characterized by a finite velocity of propagation of perturbations, known as the Lieb-Robinson velocity. On the other hand, irreducible stochastic processes drive states towards some unique fixed point. However, combining both effects is mathematically challenging. The bounds on propagation do not depend on system-size, while the theory of mixing is mostly based on extensive upper-bounds.
In \cite{Descamps}, a class of local Lindbladian operators on arbitrary lattice was constructed, for which the two effects could be combined.
In this paper, we show that for some local dynamics, local observables are propagated inside a much smaller convex subset of the total Banach space. This allows us to show localization for such dynamics.
\end{abstract}

\maketitle

\section{Introduction}
In 1958, Anderson \cite{Anderson} showed that random impurity potentials produce the localization of quantum wavefunctions. This  lead to a better understanding of transitions between insulators and conductors.
In quantum spin-lattice, it can be shown that due to (quasi)-local interactions perturbations propagate at a finite velocity. This was proven by Lieb and Robinson\cite{LIEB}. Recently, there has been an increasing interest in studying local dissipative  dynamics. One of the major drive behind, is the possibility of using dissipation as a tool for engineering states. Similarly to classical Monte-Carlo algorithm, the goal of course is to drive any state as fast as possible towards a desired state. 
When introducing local dissipation to the system, the propagation of the perturbation competes with the global relaxation of the system.
In \cite{2007PhRvL..99p7201B, BOE}, a study of such competition was first attempted. In \cite{Descamps}, it was shown that this competition could  not only effectively reduce the velocity of propagation but also lead to some increasing deceleration. All the results so far were studied for systems with the uniform state as unique fixed point.
For some development of quantum technologies, localization of excitations on pure states such as graph state is of general interests \cite{citeulike:5219312}.

\subsection{Local Quantum Stochastic Dynamics}
Consider a $D$-dimensional lattice $\mathbb{Z}^D$ with a metric. At each point $x\in \mathbb{Z}^D$ of the lattice, define a $d$-dimensional Hilbert space $\mathcal{H}_x$ and for each finite set $\Lambda \subset  \mathbb{Z}^D$ denote the product space,
$$\mathcal{H}_\Lambda = \bigotimes\limits_{x\in \Lambda}\mathcal{H}_x$$
Denote the algebra of all matrices acting on $\mathcal{H}_\Lambda$, i.e. the algebra of local observables of $\Lambda$, $\mathcal{A}_\Lambda$,
$$\mathcal{A}_\Lambda = \bigotimes\limits_{x\in \Lambda}\mathcal{B}(\mathcal{H}_x)$$ 
If $\Lambda_1 \subset \Lambda_2$, the algebra $\mathcal{A}_{\Lambda_1}$ can be identified with the algebra $\mathcal{A}_{\Lambda_1}\otimes \mathbb{1}_{\Lambda_1 \backslash \Lambda_2}$ and therefore  $\mathcal{A}_{\Lambda_1} \subset \mathcal{A}_{\Lambda_2}$.
Define the support of a local observable $A\in \mathcal{A}_\Lambda$ as the minimal set $X\subset \Lambda$ for which $A=A'\otimes \mathbb{1}_{\Lambda\backslash X}$.

The dynamics of the systems is generated by some local Lindbladian,
\begin{align*}
\partial_t A(t)=L [A(t)],~~L[.]=\sum_Z I_Z[.]=\sum_Z R_Z^\dagger[.]R_Z+Q_Z[.]+[.]Q_Z^\dagger,~~\sum_Z R_Z^\dagger R_Z+Q_Z+Q_Z^\dagger=0,~~Q_Z\in \mathcal{A}_Z
\end{align*}
where the last condition insures that the dynamics preserves unity.

\subsection{From local Response to Lieb-Robinson Velocity}
By Frobenius-Perron, the adjoint of the dynamics has some positive fixed point $\rho_{\mbox{{\scriptsize eq}}}$,
$$L_0^{\mbox{{\scriptsize ad}}}[\rho_{\mbox{{\scriptsize eq}}}]=0,~~\rho_{\mbox{{\scriptsize eq}}}\geq0, \operatorname{Tr}\rho_{\mbox{{\scriptsize eq}}}=1$$
Let us now perturb the dynamics at some local subset $Z$, be it time-independently or not,
$$\partial_t A(t)=L_0[A(t)]+\mathcal{\phi}_Z(t) [A(t)]$$ 
We wish then to compare the expectation value of the undisturbed system $\langle A(t)\rangle=\langle A\rangle$ with the new one, $\langle A(t)\rangle_\phi$.
This is readily derived,
$$\langle A(t)\rangle_\phi-\langle A\rangle=\int_0^t \operatorname{Tr}\left(\rho_{\mbox{{\scriptsize eq}}}\phi_Z[A(s)]\right)$$
Hence the response of the systems, is determined by bounding the term inside the integral. Since $\operatorname{Tr}(\rho_{\mbox{{\scriptsize eq}})}=1$, this becomes,
$$\operatorname{Tr}\left(\rho_{\mbox{{\scriptsize eq}}}\phi_Z[A(t)]\right)\leq \|\phi_Z[A(s)]\|$$
Such expression was studied for the first time by Lieb and Robinson and by reducing the problem to a Laplacian inequality it was shown that,
\begin{align}\|\phi_Z[A(s)]\|\leq C \exp\left(\frac{vt-d(Z,A)}{\xi}\right)\label{LRB}\end{align}
where $C,v$ and $\xi$ are finite constants.

\subsection{Killing of Transport}
In \cite{Descamps}, it was shown that under some symmetry a class of Lindbladians could be constructed so that the bound takes the form,
\begin{align}
\|\phi_Z[A(s)]\|\leq C \exp\left(\frac{\int_0^t \exp(-\mu s)v\  ds-d(Z,A)}{\xi}\right)\label{LRBDesc}\end{align}
As pointed out, if one could show,
\begin{align}\max_{X}\frac{\|\exp\left(tL^c\right)[X]-\lim_{t\to\infty}\exp\left(tL^c\right)[X]\|}{\|X\|}\leq C' \exp(-\lambda t) \label{decay}
\end{align}
where $\mathcal{L}^c$ is the full generator of the dynamics minus some local interaction terms, and $C'$ is a finite constant, the bound can be shown to be of the form,
\begin{align}
\|\phi_Z[A(s)]\|\leq C'' \exp\left(\frac{vt-\lambda t-d(Z,A)}{\xi}\right)\label{LRB2}\end{align}
This becomes particularly interesting when $\lambda\approx v$, as this imply that $A(t)$ is localized inside a fixed volume for all $t$.
Proving the case for which an inequality such as (\ref{decay}) is satisfied is hard. 
We present here an alternative. We show a result similar to (\ref{decay}), by relaxing the set of (global) observables to a smaller convex subset. Under the additional property that the dynamics preserve the convex structure, a bound such as (\ref{LRB2}) can then be proven, and even more so localization can be implied.

\section{Localization}
Consider some local basis $\{\sigma^{(j)}_{k}\}$ of the local Hilbert space $\mathbb{C}^d_{(j)}$. For the lattice $\otimes_{j\in\Lambda}\mathbb{C}^d_{(j)}$, define the convex hull of the set constructed from the local basis multiplied by some phase $\mathcal{C}=\operatorname{Conv}\left(\{e^{i\phi}\otimes_j \sigma^{(j)}_{k_j}|\phi\in [0,2\pi]\}\right)$.
We say that this set is normalizable with local observables, if $\forall A\in \mathcal{A}_{\Lambda_A}$, there is some finite constant $0<C_A<\infty$, so that, $A/C_A \in \mathcal{C}$.

We say that a local dynamics preserves $\mathcal{C}$ if, 
\begin{align}
\forall t, X\in \mathcal{C}, \exists Y(t) \in \mathcal{C}, 0\geq r(t)\leq 1,~~\exp(tL)[X]=r(t) Y(t) \label{structure}
\end{align}

As, we see further on, it is then sufficient to restrict (\ref{decay}) to the inequality,
\begin{align}\max_{X\in \mathcal{C},\|X\|=1}\frac{\|\exp\left(tL^c\right)[X]-\lim_{t\to\infty}\exp\left(tL^c\right)[X]\|}{\|X\|}\leq \exp(-\lambda t) \label{decay2}
\end{align}
where $\mathcal{L}^c$ denotes the global generators minus any number of local interactions $I_Z$.
As usual in order to take into account long range interaction, we introduce the usual formalism of reproducing functions.
We assume there exists a non-increasing function $F:[0,\infty)\to (0,\infty)$ so that $F$ is uniformly integrable over the lattice,

$$\| F \| := \underset{x\in \mathbb{Z}^d}{\operatorname{sup}}\sum\limits_{y\in \mathbb{Z}^d}F(d(x,y))<\infty$$

and there is a constant $C_{\mu}$,

$$C_\mu := \underset{x,y\in \mathbb{Z}^d}{\operatorname{sup}}\sum\limits_{z\in \mathbb{Z}^d}\frac{F(d(x,y))F(d(y,z)}{F(d(x,z))}e^{-\mu(d(x,y)+d(y,z)-d(x,z))}$$
so that $C_{\mu}\big|_{\mu=0}<\infty$.

Assume there exists a $\mu>0$, so that,

\begin{align*}\|L\|_{\mu}= \underset{x,y\in \mathbb{Z}^d}{\operatorname{sup}}\sum\limits_{Z \ni x,y} \frac{\|I_Z\|_{cb}}{F(d(x,y))} e^{\mu d(x,y)}<\infty
\end{align*}

Finally, let us assume there exists some finite constant $0<C_{phi}<\infty$ so that $\forall Z$, $\forall X\in \mathcal{C}$,
$$L_Z[X]/C_{X}\in \mathcal{C},~~C_{X}\leq C_{\phi} \forall X $$
We say that the system is frustration free if,
$$\forall Z,~~ \mathcal{L}_Z\circ \lim_{t\to\infty}\exp(t\mathcal{L})=0$$

In order to derive localization, the dynamics needs to be split in two other local one. The first one is responsible for the spreading the perturbation, while the second is the main dissipative part which erases the perturbation.
Clearly, for localization to be possible, perturbation needs to be driven towards some zero-correlated subspace. Hence, we assume that the local parts of $\mathcal{L}_1=\sum_Z I^1_Z$ mutually commute.
For the sake of the fluidity of the argument, the interaction terms of both generators are grouped inside some finite subset $Z$. 
$$L=L_0+L_1=\sum_Z (I^{0}_Z+ I^{1}_{Z})$$

We can now state, the theorem.
\begin{result}
  Let $L=L_0+L_1=\sum_Z (I^{0}_Z+ I^{1}_{Z})$ be frustration free  and irreducible local Lindblad generator satisfying (\ref{structure}) and(\ref{decay2}), and $[I^1_Z,I^1_{Z'}]=0$.Then $\forall A\in \mathcal{A}_{\Lambda_A}, B\in \mathcal{A}_{\Lambda_B}$ with $\Lambda_A \cap \Lambda_B=\emptyset$,
\begin{equation}
\|[B,\exp(tL)[A]]\|\leq 2 \frac{\|F\|}{C_{\mu}}\|A\|\ \|B\|_{cb}\ \operatorname{min}\{\Lambda_A,\Lambda_B\}\left(\exp\left[vt-\lambda t\right]-1\right)\exp \left[ -  \mu d(A,B)\right]
\label{eventh}
\end{equation} 
with $v=C_{\phi}\|L^{0}\|_{\mu}$, and,
\begin{align}\max_{X\in \mathcal{C},\|X\|=1}\frac{\|\exp\left(tL_1\right)[X]-\lim_{t\to\infty}\exp\left(tL_1\right)[X]\|}{\|X\|}\leq \exp(-\lambda t) \label{decay2}
\end{align}

\end{result}

Hence, we see that the velocity and the decay rate are decoupled from one another. By varying the way between the diffusive part $\mathcal{L}^0$ an the dissipative part $\mathcal{L}^1$, it is then possible to get $\lambda>v$, from which localization follows.

As the proof of the result is extremely similar to the proof presented in the "further discussion" of \cite{Descamps}, we leave the proof in the appendix.

\subsection{Classical Stochastic Processes}
It is straightforward to embed a classical stochastic dynamics into a quantum one.
For example in the case of a local stochastic generator $L=\sum_Z T_Z-I_Z$, where $I_Z$ is the identity operator and $T_Z$ is local stochastic. Hence, this can be embedded into a Lindbladian using the local standard basis $\left\{e(Z)_j\right\}$,
\begin{align}\mathcal{L}[.]=\sum_Z \sum_{j} \sqrt{\operatorname{Tr}(e(Z)_jT_Z)}~~ e(Z)_j[.]e(Z)_j^\dagger-[.]_Z
\label{classstoch}
\end{align}
In the case of local Hilbert spaces $\mathbb{C}^2_j$, $e(Z)_{ij}=|i\rangle\langle j|$. Additionally, we can see a clear decoupling between the vector spaces $\operatorname{span}\{\otimes \sigma_j|\sigma_j\in\{\mathbb{1},\sigma_z\}\}$ and $\operatorname{span}\{\otimes \sigma_j|\sigma_j\in\{\mathbb{1},\sigma_x, \sigma_y\}\}$.
Hence for simplicity, let us restrict to dynamics on lattices $\otimes_{j\in \Lambda}\mathbb{C}^2$, and dynamics (\ref{classstoch}) restricted to observables in $\operatorname{span}\{\otimes \sigma_j|\sigma_j\in\{\mathbb{1},\sigma_z\}\}$

Define now the local unitary $U_j=\left(\sigma_z+\sigma_x\right)/\sqrt{2}$. It is then also understood that $U_Z=\otimes_{j\in Z}U_j$.
In the following definition the matrix $|A|_{ij}$ is understood as the entry-wise absolute value. 
The following proposition gives us a local condition that implies property  (\ref{structure}).
\begin{proposition}
If the matrix $|U_Z^\dagger T_Z^\dagger U_Z|_{ij}$ is sub-stochastic then the dynamics (\ref{classstoch}) preserves the convex set $\operatorname{span}\{e^{i\phi}\otimes \sigma_j|\sigma_j\in\{\mathbb{1},\sigma_z\},\phi\in[0,2\pi]\}$ as defined in  (\ref{structure}).
\end{proposition}
\begin{proof}
Since,
$$\exp(A+B)=\lim_{n\to\infty}\left((I+A/n)(I+B/n)\right)^n$$
and due to the definition of (\ref{classstoch}), we can decompose each interation term as $(1-\epsilon) I +\epsilon T$. Hence, it is sufficient to show that, $\sqrt{\operatorname{Tr}(e(Z)_jT_Z)}~~ e(Z)_j[.]e(Z)_j^\dagger$ satisfy the claim.
It can be seen that the claim is equivalent to showing that $\forall \tau \in \{\otimes \sigma_j|\sigma_j\in\{\mathbb{1},\sigma_z\}\}$ ,
$$\sum_{\sigma \in\{\otimes \sigma_j|\sigma_j\in\{\mathbb{1},\sigma_z\}\}}|\langle \sigma|T_Z|\tau\rangle|\leq 2^{|Z|}$$
which is indeed equivalent to $|U_Z^\dagger T_Z^\dagger U_Z|_{ij}$ being sub-stochastic.
\end{proof}
As an illustrative example, we can consider the dynamics,
$$\mathcal{L}[.]=\mathcal{L}^{0}[.]+\lambda \left(\frac{1}{2}\sum_j \operatorname{Tr}_j[.]-[.]\right)$$
Clearly we can see,
\begin{align*}
\max_{\mathcal{C},\|X\|=1}\|\exp\left(tL^c\right)[X]-\lim_{t\to\infty}\exp\left(tL^c\right)[X]\|&\leq \lim_{n\to\infty}\|(1-t\lambda/n)[X]+t\lambda/2n \sum_j \operatorname{Tr}_j[X]-P[X]]\|^n\\
&\leq \exp(-\lambda t)
\end{align*}
with $P[X]=\left(\frac{1}{2}\right)^{|\Lambda|}$.

\subsection{Engineering dynamics of Graph States}
As a second and final illustration, we show that some Lindbladians  known to generate graph states satisfy the properties (\ref{decay2}) and (\ref{structure}) for $\mathcal{C}=\operatorname{span}\{e^{i\phi}\otimes \sigma_j|\sigma_j\in\{\mathbb{1},\sigma_x,\sigma_y,\sigma_z\},\phi\in[0,2\pi]\}$.
Graph states are defined as the unique eigenstate of the set of commuting observables $U_k=\sigma_{x,k} \otimes_{j\in Z}\sigma_{z,j}$, with $Z$ some neighborhood $k$. As shown in \cite{Kraus}, defining $c^{\alpha}_k=\frac{1}{2}(I+U_k)\sigma_{x,\alpha}$ with $\alpha\in\{x,y\}$, the graph state can be prepared by each local Lindbladian,
$$\mathcal{L}^{\alpha}[.]=\sum_k c_k^{\alpha\dagger}[.] c_k^{\alpha}-\frac{1}{2}\{ c_k^{\alpha\dagger} c_k^{\alpha},.\},~~\alpha\in\{x,y,z\}$$
Notice that $P_k=\frac{1}{2}(I+U_k)$ is a projector. 
Let us therefore take $\Sigma \in \{\otimes \sigma_j|\sigma_j\in\{\mathbb{1},\sigma_x,\sigma_y,\sigma_z\}\}$.
Let us fix $\alpha=x$ as the other case are proven similarly.
We need then to distinguish between the cases $\sigma_{\alpha,k} \Sigma=\pm \Sigma \sigma_{\alpha,k}$ and $U_k \Sigma=\pm \Sigma U_k$ and calculate for each,
\begin{enumerate}
\item  $\sigma_{\alpha,k} \Sigma=\Sigma \sigma_{\alpha,k},~U_k \Sigma=\Sigma U_k$ implies that $\mathcal{L}^{x}_k[\Sigma]=0$ and so $\exp(t\mathcal{L}^{x}_k)[\Sigma]=\Sigma$.

\item  $\sigma_{\alpha,k} \Sigma=\Sigma \sigma_{\alpha,k},~U_k \Sigma=-\Sigma U_k$
 gives $\mathcal{L}^{x}_k[\Sigma]=-\Sigma$ and so, $\exp(t\mathcal{L}^{x}_k)[\Sigma]=e^{-\lambda t}\Sigma$

\item $\sigma_{\alpha,k} \Sigma=-\Sigma \sigma_{\alpha,k},~U_k \Sigma=\Sigma U_k$ yields $\mathcal{L}^{x}_k[\Sigma]=-\Sigma-\Sigma \sigma_{\alpha,k}U_k\sigma_{\alpha,k}$ and so, $\exp(t\mathcal{L}^{x}_k)[\Sigma]=\Sigma P_k+e^{-2\lambda t}\Sigma (\mathbb{1}-P_k)$

\item $\sigma_{\alpha,k} \Sigma=-\Sigma \sigma_{\alpha,k},~U_k \Sigma=-\Sigma U_k$ also gives $\mathcal{L}^{x}_k[\Sigma]=-\Sigma$ and so, $\exp(t\mathcal{L}^{x}_k)[\Sigma]=e^{-\lambda t}\Sigma$
\end{enumerate}
Since $\sigma_{\alpha,k}U_k\sigma_{\alpha,k} \in \mathcal{C}$, all cases therefore imply (\ref{structure}).
Finally in order to prove (\ref{decay2}), since per construction the Lindbladian has a dimensional one kernel, case 1 cannot be satisfied for all $k$. The claim also follows if either case 2 or 4 are fulfilled for a single $k$. Hence, the only remaining point to look at, is when case 3 is satisfied at least once. So for all $k$, $[\Sigma,P_k]=0$, and we can take mutual eigenvectors $|\alpha\rangle$,  $\Sigma|\alpha\rangle=s_{\alpha}|\alpha\rangle$, $s\alpha=\pm1$ and $P_k |\alpha\rangle= p_{\alpha,k}|\alpha\rangle$, $p_{\alpha,k}\in\{0,1\}$. Clearly, since the graph state is per definition the unit vector which has eigenvalue $1$ for all $P_k$, it has to be an eigenvector of $\Sigma$ too. Hence for all $|\alpha\rangle $ orthogonal to the graph state, we have,
$$|\langle \alpha|\exp(t\mathcal{L}^{\alpha})[\Sigma]-\lim_{t\to\infty}\exp(t\mathcal{L}^{\alpha})[\Sigma]|\alpha\rangle|\leq e^{-2t}$$
Hence all case combined, we see that the decaying condition (\ref{decay2}) is satisfied.

\section{Conclusion}
We studied a class of local Lindbladians which preserves a particular local convex structure. We show that in the irreducible case, a competition can emerge between the diffusive and strongly mixing part of the generator. We see that the diffusion property can be decoupled from the mixing. Localization is then implied for a finite mixing rate of the same order as the Lieb Robinson velocity.
At the end, we showed various construction from classical stochastic processes, to more complex quantum dynamics with graph states as fixed point.
\appendix
\section{Proof of Result 1}
\begin{proof}
The proof is very similar to the proofs presented in [REF DESCAMPS]
Define a path of  $n$ steps, $\mathcal{P}_n$, between $A$ and $B$ as the set of the support of the local interactions $I_{\Lambda_j}$,
$$\mathcal{P}_n=\{\Lambda_j|\Lambda_{B} \cap \Lambda_1 \not=\emptyset , \Lambda_1 \cap \Lambda_2 \not=\emptyset,...,\Lambda_n \cap \Lambda_{A}  \not=\emptyset \mbox{ , all other intersections are empty}\}$$ 
Denote $\mathcal{P}_n^k $ as the subset of 
$\mathcal{P}_n$ containing the first $k$ steps of the path, where $\mathcal{P}_n^0 =\emptyset $. 
Define the generator part $\mathcal{L}^0$ containing local interaction which do not intersect with the path,
$$L^c_{\mathcal{P}_n^k}= L-\sum_{Z\cap \mathcal{P}_n^k\not= \emptyset } I^{0}_Z $$

In Lemma 1 we show,
\begin{equation*}
\|[B, \exp(tL)A]\| \leq  \sum_{n=1}^\infty \|L\|_\mu^n \left(\sup\limits_{\mathcal{P}_n}\mathcal{J}_{\mathcal{P}_n}\right) \sum\limits_{\mathcal{P}_n}\sum\limits_{x_0\in \Lambda_B}\prod_{\Lambda_j \in \mathcal{P}_n}\sum_{x_j \in \Lambda_j}\sum\limits_{x_{n+1} \in \Lambda_A}F(d(x_{j-1},x_j)) e^{-\mu d(x_{j-1},x_j)}
\end{equation*}
with,
\begin{align}
&\mathcal{J}_{\mathcal{P}_n} = \Big|\Big|\int_0^t \int_0^{t-s_1
}...\int_0^{t-\sum_{j=1}^{n-1}s_j} d\textbf{s}\ [B,.]  \  \prod_{j=1}^{n}\left(\exp\left[s_j L^c_{\mathcal{P}_n^{j-1}}\right]  I^{0}_{\Lambda_j}/\|I^{0}_{\Lambda_j}\|_{cb}\right) \exp[(t-\sum_{j=1}^{n}s_j) L^c_{\mathcal{P}_n^n}]A \Big|\Big|
\label{Jpn}
\end{align}
with the ordered product $\prod_{j=1}^n C_j=C_1 C_2 ... C_n$ and the supremum is taken over all paths of length $n$.

Define the projector,
$$P=\lim_{t\to\infty}\exp(t\mathcal{L}^1)$$

The frustration free and irreducible condition allow us to insert the projectors,

\begin{align*}
\mathcal{J}_{\mathcal{P}_n} = &\Big|\Big|\int_0^t \int_0^{t-s_1
}...\int_0^{t-\sum_{j=1}^{n-1}s_j} d\textbf{s}\ [B,.]  \  \prod_{j=1}^{n}\left(\exp\left[s_j L^c_{\mathcal{P}_n^{j-1}}\right](\mathbb{1}-P)  I^{0}_{\Lambda_j}/\|I^{0}_{\Lambda_j}\|_{cb}\right)\\ &\exp[(t-\sum_{j=1}^{n}s_j) L^c_{\mathcal{P}_n^n}](\mathbb{1}-P)A \Big|\Big|
\end{align*}

Since, $[L^c_{\mathcal{P}_n^k},P^c_{\mathcal{P}_n^{k}}]=0$, $\forall \mathcal{P}_n,k$ and using the assumption (\ref{decay2}), it follows,
$$\mathcal{J}_{\mathcal{P}_n}\leq \exp(-\lambda t) \frac{C_{\phi}^n\|L^{0}\|_{\mu}^n t^n}{n!}$$

From which follows,
$$\|[A, \Gamma_t(B)] \|\leq C' \|A\|\ \|B\| \exp\left( \frac{-\lambda t + v t - d(A,B)}{\xi} \right)$$

\end{proof}
\begin{Lemma}
For  $A$ and $I_Z$ with $Z \cap \Lambda_A =\emptyset$,
\begin{equation*}
\|\ I^0_Z \exp(tL)A\ \| \leq \sum_{n=1}^\infty \|L^0\|_\mu^n \left(\sup\limits_{\mathcal{P}_n}\mathcal{J}_{\mathcal{P}_n}\right) \sum\limits_{\mathcal{P}_n}\sum\limits_{x_0\in \Lambda_A}\prod_{\Lambda_j \in \mathcal{P}_n}\sum_{x_j \in \Lambda_j}\sum\limits_{x_{n+1} \in Z}F(d(x_{j-1},x_j)) e^{-\mu d(x_{j-1},x_j)}
\end{equation*}
\end{Lemma}
\begin{proof}
The idea of the proof, is the gradually take out, each interaction term of the diffusive generator $\mathcal{L}^{0}$ forming so a path from $Z$ to $A$.
 Define the generator $L^{0}_{\cap Z}$,
$$L^{0}_{\cap Z} =\sum_{\Lambda_1\cap Z\not=\emptyset}I^{0}_{\Lambda_1} $$
We can then expand,

$$I^{0}_Z \exp(tL)A = \sum_{\Lambda_1\cap Z\not=\emptyset}\|I_{\Lambda_1}\|_{cb}\int_0^t I^{0}_Z\exp[s_0 L] I^{0}_{\Lambda_1}/\|I_{\Lambda_1}\|_{cb} \exp[(t-s_0)(L-L^{0}_{\cap Z})]A$$
We then continue this procedure for each term $I_{\Lambda_1} \exp[(t-s_0)(L-L^{0}_{\cap Z})]A$  for which $\Lambda_1 \cap \Lambda_A =\emptyset$. Let us say that after $k-1$ steps we have $\Lambda_k \cap \Lambda_A \not=\emptyset$, from our construction we see that we have built a path of $k$ steps from $Z$ to $A$ for which $\mathcal{J}_{\mathcal{P}_k(Z,A)}$ can be bounded by equation (\ref{Jpn}). For each $n$, we can then use the definition of $\|L\|_{\mu}$,
\begin{align*}
\left(\sup\limits_{\mathcal{P}_n}\mathcal{J}_{\mathcal{P}_n}\right)\sum_{\mathcal{P}_n}\prod_{\Lambda_j \in \mathcal{P}_n}\sum_{\Lambda_j}\|I^0_{\Lambda_j}\|_{cb}\leq \left(\sup\limits_{\mathcal{P}_n}\mathcal{J}_{\mathcal{P}_n}\right) \sum\limits_{\mathcal{P}_n} \|L^0\|_{\mu}^n \sup\limits_{x_0\in \Lambda_A}\prod_{\Lambda_j \in \mathcal{P}_n}\sum_{x_j \in \Lambda_j}\sum\limits_{x_{n+1} \in Z}F(d(x_{j-1},x_j))  e^{-\mu d(x_{j-1},x_j)}
\end{align*}
\end{proof}

\end{document}